\documentclass[journal,twoside,web,final]{ieeecolor}

\let\labelindent\relax

\makeatletter
\let\NAT@parse\undefined
\makeatother

\def\endfigure{\end@float}
\def\endtable{\end@float}

\usepackage{amsfonts}
\usepackage{amsmath}
\usepackage{amssymb}
\usepackage{amsthm}
\usepackage{accents} 
\usepackage{bm}
\usepackage{cite}
\usepackage{enumitem}
\usepackage{graphicx}
\usepackage{mathtools}
\usepackage[caption=false]{subfig}
\usepackage{textcomp}
\usepackage{tikz}
\usepackage{xurl}
\usepackage[hidelinks]{hyperref} 
\usepackage{cleveref} 

\def\BibTeX{{\rm B\kern-.05em{\sc i\kern-.025em b}\kern-.08em
    T\kern-.1667em\lower.7ex\hbox{E}\kern-.125emX}}

\definecolor{subsectioncolor}{rgb}{0.067,0.627,0.859}
\DeclareMathOperator*{\minimize}{minimize\quad}
\DeclareMathOperator*{\subjectto}{subject\ to\quad}
\DeclareMathOperator*{\tr}{tr}

\theoremstyle{definition}
\newtheorem{assumption}{Assumption}
\theoremstyle{definition}
\newtheorem{definition}{Definition}
\theoremstyle{plain}
\newtheorem{lemma}{Lemma}
\theoremstyle{plain}
\newtheorem{problem}{Problem}
\theoremstyle{plain}
\newtheorem{proposition}{Proposition}
\theoremstyle{remark}

\theoremstyle{plain}
\newtheorem{theorem}{Theorem}

\crefname{assumption}{}{}
\Crefname{assumption}{Assumption}{Assumptions}
\crefname{equation}{}{}
\Crefname{equation}{Equation}{Equations}
\crefname{figure}{}{}
\Crefname{figure}{Fig.}{Figs.}
\crefname{problem}{}{}
\Crefname{problem}{Problem}{Problems}

\newlist{assumptionenum}{enumerate}{1}
\setlist[assumptionenum]{label=(\alph*), ref=\theassumption(\alph*)}
\crefalias{assumptionenumi}{assumption}

\allowdisplaybreaks

\title{
    \LARGE Risk-Aware Finite-Horizon Social Optimal Control of \\ Mean-Field Coupled Linear-Quadratic Subsystems
}

\author{
    Dhairya~Patel and Margaret~P.~Chapman,~\IEEEmembership{Member,~IEEE}%
    \vspace{-10mm}%
    \thanks{
        This research is supported by the Edward S. Rogers Sr. Department of Electrical and Computer Engineering (ECE), University of Toronto, and the Natural Sciences and Engineering Research Council of Canada (NSERC) Discovery Grants Program, [RGPIN-2022-04140]. Cette recherche a \'{e}t\'{e} financée par le Conseil de recherches en sciences naturelles et en g\'{e}nie du Canada (CRSNG).
    }%
    \thanks{
        D.P. and M.P.C. are with the Edward S. Rogers Sr. Department of ECE, University of Toronto, 10 King's College Road, Toronto, Ontario M5S 3G8 Canada (e-mail: \url{dhairya.patel@mail.utoronto.ca}).
    }
}

\begin{document}

\maketitle

\thispagestyle{empty}

\begin{tikzpicture}[remember picture, overlay]
    \node[anchor = south, yshift = 2.5pt] at (current page.south) {%
        \fbox{%
            \parbox{%
                \dimexpr0.9\textwidth-\fboxsep-\fboxrule\relax%
            }{%
                \footnotesize \textcopyright \ \the\year{} IEEE. Personal use of this material is permitted. Permission from IEEE must be obtained for all other uses, in any current or future media, including reprinting/republishing this material for advertising or promotional purposes, creating new collective works, for resale or redistribution to servers or lists, or reuse of any copyrighted component of this work in other works.
            }%
        }%
    };
\end{tikzpicture}
\vspace{-\baselineskip}

\begin{abstract}
    \hbadness=10000
    We formulate and solve an optimal control problem with cooperative, mean-field coupled linear-quadratic subsystems and additional \emph{risk-aware} costs depending on the covariance and skew of the disturbance. This problem quantifies the variability of the subsystem state energy rather than merely its expectation. In contrast to related work, we develop an alternative approach that illuminates a family of matrices with many analytical properties, which are useful for effectively extracting the mean-field coupled solution from a standard LQR solution.
\end{abstract}

\begin{IEEEkeywords}
    \hbadness=10000
    Cooperative control, linear systems, stochastic optimal control.
    \vspace{-3mm}
\end{IEEEkeywords}

\bstctlcite{IEEEbstcontrol}

\section{Introduction}
\label{sec:introduction}

\IEEEPARstart{T}{eam} theory concerns the control of complex systems consisting of many cooperating subsystems \cite[Def. 18.1]{vanSchuppen2015coordination}. This has diverse applications, such as coordinating fleets of autonomous vehicles \cite{guney2018schedulingdriven,liu2006motion}, signalling and traffic \cite{boel2015leaderfollower}, and sensor and surveillance networks \cite{phoompat2020coordinated,kucevic2021reducing}. In such settings, we may require a controller with minimal information sharing between subsystems for practical reasons, including limited communication bandwidth or computational power. This is difficult in general, so we study a model where the \emph{mean-field}---the mean of all subsystem states---turns out to be the only information shared, which can be implemented in practice using a simple microcontroller that receives each subsystem state, then sends back the computed mean. Furthermore, we focus on a stochastic model with individual linear dynamics and a common quadratic objective, leading to the linear-quadratic regulator (LQR) problem for \emph{mean-field coupled subsystems}, and a corresponding \emph{social} optimal control solution.\footnote{Stochastic subsystems interacting through their mean-field have been of recent interest, primarily in the alternative mean-field \emph{games} setting, where one finds the \emph{Nash} optimal control solution for noncooperative subsystems optimizing individual objectives, e.g., see \cite{başar2021robust,jun2014discrete,uzZaman2020reinforcement,aydın2023robustness}. Applications for mean-field games can correspond to applications for mean-field coupled social optimal control when the subsystems are assumed to be cooperative instead, e.g., see \cite{nourian2013nash}.}

Social optimal controls for mean-field coupled subsystems have been developed in continuous- and discrete-time linear-quadratic settings \cite{wang2021indefinite,du2022social,arabneydi2015teamoptimal}. The continuous-time approaches in \cite{wang2021indefinite,du2022social} consider only Wiener process disturbances, limiting the systems to which they may be applied. The approaches in \cite{wang2021indefinite,du2022social,arabneydi2015teamoptimal} are \emph{risk-neutral} as they consider only the expected standard cost, thereby limiting their ability to avoid or compensate for dangerous, but possibly infrequent, disturbances with catastrophic consequences \cite{wang2022risk}. This motivates the formulation of a \emph{risk-aware} problem that incorporates further information about the distribution of possible outcomes to mitigate high-risk scenarios.

While the field of control systems traditionally has focused on analysis and design in the average case or in the worst case, \emph{risk-aware control theory} is the study of control systems with respect to diverse characterizations of future possibilities between the average case and the worst case, offering enhanced flexibility compared to the standard paradigms \cite{wang2022risk}.\footnote{While risk-aware control is a main aspect of this work, distributionally robust control (e.g., optimizing a worst-case expectation) also does not merely rely on the expected standard cost to quantify uncertain outcomes. Distributionally robust risk-aware control has been studied as well \cite{vanParys2016distributionally}.} There are many methods of equipping control systems with risk awareness. For example, exponential-of-cost penalizes increases in the objective more strongly to discourage larger costs, however it typically requires Gaussian disturbances and its risk parameter may have a limited effective range \cite{whittle1990risk,smith2023on}. Worst-case approaches may be too pessimistic and cannot handle noise distributions with noncompact support. Conditional-value-at-risk-based approaches use a risk parameter that quantifies a fraction of the worst outcomes, thereby offering more flexibility \cite{wang2022risk,chapman2021on,chapman2022risk,vanParys2016distributionally}. Predictive variance assesses risk in terms of variability \cite{tsiamis2020riskconstrained}, reflecting a classical risk assessment approach \cite{markowitz1952portfolio}, and may benefit applications that make decisions based on expectation-variability trade-offs. In a linear-quadratic setting with mild restrictions on the disturbance (e.g., the disturbance does not need to be Gaussian), predictive variance admits a quadratic form, leading to an optimal controller that depends on the covariance and skew of the disturbance and can attenuate drastic changes in the state energy effectively \cite{tsiamis2020riskconstrained,chapman2022riskaware}.

\subsubsection*{Contribution}

We formulate and solve a risk-aware social optimal control problem for mean-field coupled linear-quadratic subsystems on a finite horizon under mild assumptions on the process noise by adapting the predictive variance in \cite{tsiamis2020riskconstrained} to introduce a risk-aware cost for each subsystem. A closely related work to ours is \cite{roudneshin2023risk}, which considers an infinite-horizon, risk-constrained formulation instead. However, our work differs from and improves upon \cite{roudneshin2023risk} in crucial ways. Rather than using the common reformulation based on auxiliary and mean-field states as in \cite{arabneydi2015teamoptimal,roudneshin2023risk}, we present a centralized reformulation in terms of matrices that admit a common structure, which we call \emph{pseudo-block diagonal matrices}. In contrast to \cite{roudneshin2023risk}, we explicitly and rigorously establish a family of linear maps and its many analytical properties inherited from the Kronecker product, which allows us to translate between standard uncoupled LQR and mean-field coupled settings. Pseudo-block diagonal matrix algebra is useful for adapting established theory for uncoupled systems to mean-field coupled subsystems and is not restricted to the setting of predictive variance. Once an uncoupled LQR problem with a known solution is formulated in terms of pseudo-block diagonal matrices, pseudo-block diagonal matrix algebra can enable convenient extraction of the mean-field coupled solution from the aforementioned known solution, as we demonstrate in the process from Theorem \ref{thm:centralized_control} to Theorem \ref{thm:mean_field_coupled_control} in this work, circumventing the need to rework a standard solution procedure from the start in terms of auxiliary and mean-field states as in \cite{roudneshin2023risk}. Additionally, we demonstrate that the optimal control is mean-field coupled even with the assumption of a more general information sharing structure than in \cite{roudneshin2023risk}. More broadly, we are mathematically rigorous in our problem formulation and analysis, such as being careful to condition upon the history of \emph{all} subsystems in the definition of predictive variance unlike \cite{roudneshin2023risk}, to ensure correctness in the results presented. Finally, we express the risk as a cost, rather than a constraint, bypassing the need for primal-dual machinery and associated computational complexity.

\section{Problem Formulation and Notation}
\label{sec:problem_formulation_and_notation}

\subsubsection*{Notation}

\(\mathbb{R}\) denotes the real line. \(\mathbb{R}^n\) denotes the \(n\)-dimensional real space. \(\mathbb{R}^{n \times m}\) denotes the set of \(n \times m\) real matrices. \(\mathbb{N} = \{1, 2, \ldots\}\) and \(\mathbb{N}_0 = \{0, 1, \ldots\}\) denote the sets of natural numbers and nonnegative integers, respectively. \(\mathcal{S}^n\), \(\mathcal{S}^n_+\), and \(\mathcal{S}^n_{++}\) denote the sets of symmetric, symmetric positive semidefinite, and symmetric positive definite matrices in \(\mathbb{R}^{n \times n}\), respectively. \(I_n\) denotes the \(n \times n\) identity matrix. \(0_n\) denotes the matrix of all zeros in \(\mathbb{R}^{n \times n}\). \(\mathbf{1}_n\) is the vector of all ones in \(\mathbb{R}^n\). We define \(E_n \triangleq \frac{1}{n} \mathbf{1}_n \mathbf{1}_n^\top\). For any two matrices \(M\) and \(N\), \(M \otimes N\) denotes their Kronecker product \cite[Ch. 2]{hardy2019matrix}. \(\| \cdot \|\) is the Euclidean norm on \(\mathbb{R}^n\). For any sequence of \(k \in \mathbb{N}\) column vectors \(v^1, \ldots, v^k\), we define \(\mathbf{v} = (v^1, \ldots, v^k)\) as their vertical concatenation. If \(v^1, \ldots, v^k\) have the same dimension, \(\bar{v} \triangleq \frac{1}{k} \sum_{i = 1}^k v^i\) denotes their \emph{mean-field}. Superscripts on variables denote the subsystem index, while subscripts denote the time index. \((\Omega, \mathcal{F}, \mathbb{P})\) is a probability space upon which we define all random vectors in this work, and \(\mathbb{E}(\cdot)\) is the corresponding expectation operator. If \(X_1, \dots, X_k\) is a sequence of random (column) vectors, \(\sigma(X_1, \ldots, X_k)\) denotes the \(\sigma\)-algebra generated by the random vector \((X_1, \ldots, X_k)\) \cite[Def. 6.4.1]{ash1972probability}.

\subsubsection*{Dynamics}

In this work, we consider \(k \in \mathbb{N}\) linear-quadratic (LQ) subsystems with identical dynamics and costs over a discrete-time finite horizon \(\mathbb{T} \triangleq \{0, 1, \ldots, T - 1\}\) of length \(T \in \mathbb{N}\), with mean-field coupled dynamics and costs. The random vectors \(x_t^i\), \(u_t^i\), and \(w_{t + 1}^i\) denote the \(\mathbb{R}^n\)-valued state, \(\mathbb{R}^m\)-valued control, and \(\mathbb{R}^n\)-valued disturbance of subsystem \(i \in \mathbb{I} \triangleq \{1, \ldots, k\}\) at time \(t \in \mathbb{T}\), respectively. The mean-field coupled dynamics for each subsystem are
\begin{equation}
    \label{eq:mean_field_dynamics}
    x_{t + 1}^i = A_t x_t^i + B_t u_t^i + C_t \bar{x}_t + w_{t + 1}^i, \quad i \in \mathbb{I}, \quad t \in \mathbb{T},
\end{equation}
where \(A_t \in \mathbb{R}^{n \times n}\), \(B_t \in \mathbb{R}^{n \times m}\), and \(C_t \in \mathbb{R}^{n \times n}\) are the state, input, and mean-field coupling matrices, respectively. For each \(t \in \tilde{\mathbb{T}} \triangleq \{0, 1, \ldots, T\}\), we define the random vector \(h_t \triangleq (\mathbf{x}_0, \dots, \mathbf{x}_t, \mathbf{u}_0, \dots, \mathbf{u}_{t - 1})\) and the \(\sigma\)-algebra \(\mathcal{F}_t \triangleq \sigma(h_t)\), with \(\mathcal{F}_{-1}\) the trivial \(\sigma\)-algebra \(\{\varnothing, \Omega\}\). In this work, we use the following assumptions.

\begin{assumption}[Disturbance]
    \label{assn:disturbance}\phantom{}
    \begin{assumptionenum}
        \item \label{assn:disturbance_a} The random vectors \(\mathbf{x}_0, \mathbf{w}_1, \ldots, \mathbf{w}_T\) are independent, \(\mathbf{w}_1, \ldots, \mathbf{w}_T\) are i.i.d., and \(\mathbf{x}_0\) is deterministic.
        \item \label{assn:disturbance_b} \(w_{t + 1}^1, \ldots, w_{t + 1}^k\) are i.i.d. for each \(t \in \mathbb{T}\).
        \item \label{assn:disturbance_c} \(\mathbb{E}(\lVert w_{t + 1}^i \rVert^4) < \infty\) for each \(i \in \mathbb{I}\) and \(t \in \mathbb{T}\).
    \end{assumptionenum}
\end{assumption}

\begin{assumption}[Control]
    \label{assn:control}
    \(\mathbb{E}(\|u_t^i\|^2) < \infty\) and \(u_t^i\) is \(\mathcal{F}_t\)-measurable, denoted as \(u_t^i \in \mathcal{L}^2(\mathcal{F}_t)\), for each \(i \in \mathbb{I}\) and \(t \in \mathbb{T}\).
\end{assumption}

\Cref{assn:disturbance_a,assn:disturbance_b} are relatively standard for mean-field coupled subsystems, and we require \Cref{assn:disturbance_c} to ensure that the risk-aware objective is finite. \Cref{assn:disturbance_c} does restrict the disturbances we may consider, but still presents a significant relaxation of the common assumption of Gaussian disturbances. \Cref{assn:control} is standard and equivalent to \(\mathbb{E}(\|\mathbf{u}_t\|^2) < \infty\) with \(\mathbf{u}_t\) being \(\mathcal{F}_t\)-measurable, denoted as \(\mathbf{u}_t \in \mathcal{L}^2(\mathcal{F}_t)\), for each \(t \in \mathbb{T}\). Note that \Cref{assn:control} does not assume that \(u_t^i\) depends only on the history of the mean-field and subsystem \(i\), but rather a more general information sharing structure in which the control depends on the complete system history. We will later see that the optimal control depends only on a mean-field sharing information structure even without explicitly assuming it. Further, note that \Cref{assn:disturbance,assn:control} together with the form of the dynamics imply that \(\mathbb{E}(\|x_t^i \|^2) < \infty\) for each \(i \in \mathbb{I}\) and \(t \in \tilde{\mathbb{T}}\).

\subsubsection*{Costs}

The standard mean-field coupled per-step costs for subsystem \(i \in \mathbb{I}\) are
\begin{subequations}
    \label{eq:mean_field_costs}
    \begin{align}
        \label{eq:mean_field_cost_state}
        c_t^x(\mathbf{x}_t) &\triangleq \textstyle \sum_{i = 1}^k \bar{x}_t^\top P_t \bar{x}_t + {x_t^i}^\top Q_t x_t^i, && t \in \tilde{\mathbb{T}}, \\
        \label{eq:mean_field_cost_input}
        c_t^u(\mathbf{u}_t) &\triangleq \textstyle \sum_{i = 1}^k {u_t^i}^\top R_t u_t^i, && t \in \mathbb{T},
    \end{align}
\end{subequations}
where \(P_t \in \mathcal{S}^n_+\), \(Q_t \in \mathcal{S}^n_+\), and \(R_t \in \mathcal{S}^m_{++}\). We quantify risk-awareness using the notion of \emph{predictive variance} to develop an additional cost. The predictive variance for subsystem \(i \in \mathbb{I}\) at time \(t \in \tilde{\mathbb{T}}\) is \(\mathbb{E}({\Delta_t^i}^2)\), where the \emph{state-energy prediction error} \(\Delta_t^i\) is defined as
\begin{equation}
    \label{eq:prediction_error}
    \Delta_t^i \triangleq {x_t^i}^\top Q_t x_t^i - \mathbb{E}\left({x_t^i}^\top Q_t x_t^i \, \middle| \, \mathcal{F}_{t - 1}\right).
\end{equation}
We adapt this subsystem-specific definition from \cite{tsiamis2020riskconstrained}, which introduced an analogous notion of the predictive variance of a single uncoupled system to formulate a risk-constrained problem. Now, we can define the \emph{risk-aware mean-field coupled per-step cost} for a user-specified risk-parameter \(\lambda \geq 0\) as
\begin{equation}
    \label{eq:mean_field_cost_risk_aware}
    \textstyle c_t^\Delta(\mathbf{x}_t) \triangleq \lambda \sum_{i = 1}^k {\Delta_t^i}^2, \quad t \in \tilde{\mathbb{T}}.
\end{equation}
With these per-step costs, we formulate the finite-horizon random (cost) variable as
\begin{equation*}
    \label{eq:mean_field_objective}
    \textstyle J(\mathbf{u}) \triangleq c_T^x(\mathbf{x}_T) + c_T^\Delta(\mathbf{x}_T) + \sum_{t = 0}^{T - 1} c_t^x(\mathbf{x}_t) + c_t^\Delta(\mathbf{x}_t) + c_t^u(\mathbf{u}_t),
\end{equation*}
where \(\mathbf{u} \triangleq (\mathbf{u}_0, \mathbf{u}_1, \ldots, \mathbf{u}_{T-1})\).

\subsubsection*{Risk-Aware Optimal Control Problem}

We define the problem of interest as minimizing the expectation of \(J(\mathbf{u})\) subject to the dynamics and assumptions above. By including the predictive variance in the objective \(\mathbb{E}(J(\mathbf{u}))\) through the nonstandard costs \(c_t^\Delta\) \cref{eq:mean_field_cost_risk_aware}, this problem is risk-aware because it considers the average state energy in addition to its \emph{variability}.
\begin{problem}
    \label{prob:mean_field_optimal_control_problem}
    Under \Cref{assn:disturbance,assn:control}, the risk-aware finite-horizon optimal control problem  is
    \begin{equation}
        \label{eq:mean_field_optimal_control_problem}
        \begin{alignedat}{2}
            &\minimize  && \mathbb{E}\left(J(\mathbf{u})\right) \\
            &\subjectto && x_{t + 1}^i = A_t x_t^i + B_t u_t^i + C_t \bar{x}_t + w_{t + 1}^i, \\
            &           && u_t^i \in \mathcal{L}^2(\mathcal{F}_t), \quad i \in \mathbb{I}, \quad t \in \mathbb{T}.
        \end{alignedat}
    \end{equation}
\end{problem}

\section{Centralized Reformulation}
\label{sec:centralized_reformulation}

Toward solving \Cref{prob:mean_field_optimal_control_problem}, in this section we reformulate this problem into a centralized form that resembles typical LQR by removing explicit dependence on the mean-field. A useful observation to help derive the centralized reformulation is that the mean-field \(\bar{x}_t\) is linear in the full system state \(\mathbf{x}_t\):
\begin{equation}
    \label{eq:mean_field_reformulation}
    \textstyle \bar{x}_t = \frac{1}{k} \begin{bmatrix} I_n & \cdots & I_n \end{bmatrix} \mathbf{x}_t = \frac{1}{k} (\mathbf{1}_k^\top \otimes I_n) \mathbf{x}_t.
\end{equation}
Using this relation, we present the centralized system below.

\begin{proposition}
    \label{prop:centralized_dynamics}
    \Cref{eq:mean_field_dynamics} can be rewritten as
    \begin{equation}
        \label{eq:centralized_dynamics}
        \mathbf{x}_{t + 1} = \tilde{A}_t \mathbf{x}_t + \tilde{B}_t \mathbf{u}_t + \mathbf{w}_{t + 1}, \quad t \in \mathbb{T},
    \end{equation}
    where \(\tilde{A}_t \triangleq I_k \otimes A_t + E_k \otimes C_t\) and \(\tilde{B}_t \triangleq I_k \otimes B_t\).
\end{proposition}
\begin{proof}
    The result follows from ``vertically stacking'' the subsystem dynamics, substituting the mean-field terms using \cref{eq:mean_field_reformulation}, and then simplifying the resulting expression.
\end{proof}

\begin{proposition}
    \label{prop:centralized_costs}
    Given \(P_t \in \mathcal{S}^n_+\), \(Q_t \in \mathcal{S}^n_+\), and \(R_t \in \mathcal{S}^m_{++}\), \Cref{eq:mean_field_costs} can be rewritten as
        \label{eq:centralized_costs}
        \begin{align}
            c_t^x(\mathbf{x}_t) = \mathbf{x}_t^\top \tilde{Q}_t \mathbf{x}_t \quad \text{and} \quad
            c_t^u(\mathbf{u}_t) = \mathbf{u}_t^\top \tilde{R}_t \mathbf{u}_t,
        \end{align}
    where \(\tilde{Q}_t \triangleq I_k \otimes Q_t + E_k \otimes P_t \in \mathcal{S}^{n k}_+, \tilde{R}_t \triangleq I_k \otimes R_t \in \mathcal{S}^{m k}_{++}\).
\end{proposition}
\begin{proof}
    First, standard linear algebra results in \(c_t^u(\mathbf{u}_t) = \mathbf{u}_t^\top \left(I_k \otimes R_t\right) \mathbf{u}_t\) and \(\sum_{i = 1}^k {x_t^i}^\top Q_t x_t^i = \mathbf{x}_t^\top \left(I_k \otimes Q_t\right) \mathbf{x}_t\).
    Also, we have
    \begin{align*}
        k \bar{x}_t^\top P_t \bar{x}_t
        &=  \textstyle \frac{1}{k} ((\mathbf{1}_k^\top \otimes I_n) \mathbf{x}_t)^\top P_t ((\mathbf{1}_k^\top \otimes I_n) \mathbf{x}_t) \\
        &= \textstyle \frac{1}{k} \mathbf{x}_t^\top (\mathbf{1}_k \otimes I_n) (1 \otimes P_t) (\mathbf{1}_k^\top \otimes I_n)\mathbf{x}_t \\
        &=  \mathbf{x}_t^\top \left(E_k \otimes P_t\right) \mathbf{x}_t.
    \end{align*}

    Now, note that if \(X\) and \(Y\) are symmetric positive definite (resp., semidefinite), then \(X \otimes Y\) is symmetric positive definite (resp., semidefinite) \cite[p. 116]{hardy2019matrix}. \(I_k\) is symmetric positive definite, and we verify that \(E_k\) is symmetric positive definite in a subsequent lemma (\Cref{lemma:e_k_properties}). So, \(I_k \otimes Q_t, E_k \otimes P_t \in \mathcal{S}^{n k}_+\) and \(I_k \otimes R_t \in \mathcal{S}^{m k}_{++}\), therefore \(\tilde{Q}_t \in \mathcal{S}^{n k}_+\) and \(\tilde{R}_t \in \mathcal{S}^{m k}_{++}\).
\end{proof}

\Cref{prop:centralized_dynamics} together with \Cref{assn:disturbance,assn:control} can be used to show that \(w_{t}^i\) and \(h_{t-1}\) are independent for each \(i \in \mathbb{I}\) and \(t \in \{1, \ldots, T\}\), which is useful for reducing the risk-aware cost to a quadratic form in the following proposition.

\begin{proposition}
    \label{prop:centralized_cost_risk_aware}
    Under \Cref{assn:disturbance,assn:control}, we have
    \begin{equation*}
        \begin{aligned}
                \mathbb{E}\left(c_t^\Delta(\mathbf{x}_t)\right)
            &=  4 \lambda \mathbb{E}\left(\mathbf{x}_t^\top (I_k \otimes Q_t \Sigma Q_t) \mathbf{x}_t +
                \mathbf{x}_t^\top (\mathbf{1}_k \otimes Q_t \gamma_t)\right) \\ &\phantom{{}=} +
                k \lambda \delta_t -
                4 k \lambda \tr\left((\Sigma Q_t)^2\right),
        \end{aligned}
    \end{equation*}
    where \(c_t^\Delta(\mathbf{x}_t)\) is defined in \cref{eq:mean_field_cost_risk_aware}, \(\mu \triangleq \mathbb{E}(w_t^i)\), \(d_t^i \triangleq w_t^i - \mu\), \(\Sigma \triangleq \mathbb{E}(d_t^i d_t^{i\top})\), \(\gamma_t \triangleq \mathbb{E}(d_t^i d_t^{i\top} Q_t d_t^i)\), and
    \begin{align*}
        \delta_t \triangleq \mathbb{E}\left(\left(d_t^{i\top} Q_t d_t^i - \tr(\Sigma Q_t)\right)^2\right).
    \end{align*}
\end{proposition}
\begin{proof}
    Defining \(\ell_t \triangleq \delta_t -  4 \tr((\Sigma Q_t)^2)\), we find that
    \begin{align}
        \label{eq:predictive_variance}
            \mathbb{E}\left({\Delta_t^i}^2\right)
        &=  4 \mathbb{E}\left({x_t^i}^\top Q_t \Sigma Q_t x_t^i +
            {x_t^i}^\top Q_t \gamma_t\right) + \ell_t
    \end{align}
    using \(\mathbb{E}(\lVert \hat{x}_t^i \rVert^2) < \infty\), \(\mathbb{E}(\lVert d_t^i \rVert^4) < \infty\), and the independence of \(d_t^i\) and \(\hat{x}_t^i\), where \(\hat{x}_t^i\) is \(\mathbb{R}^n\)-valued everywhere and \(\hat{x}_t^i \stackrel{\mathrm{a.e.}}{=} \mathbb{E}(x_t^i \, | \, \mathcal{F}_{t - 1})\). The derivation for \cref{eq:predictive_variance} is involved but largely identical to the proof of \cite[Prop. 1]{tsiamis2020riskconstrained}, with minor modifications to account for the time-varying \(Q_t\) and the extra mean-field dependence in the dynamics. Then, we may use an approach similar to the one in the proof of \Cref{prop:centralized_costs} to obtain the desired result.
\end{proof}

The quadratic and affine terms in \(\mathbf{x}_t\) in the risk-aware cost have dependence on the covariance and skew of the disturbance, respectively, which we can interpret as penalizing state trajectories that are more ``sensitive'' to potential disturbances \cite{tsiamis2020riskconstrained}. Using the previous results, we determine an equivalent centralized alternative to \Cref{prob:mean_field_optimal_control_problem}.

\begin{problem}
    \label{prob:centralized_optimal_control_problem}
    Under \Cref{assn:disturbance,assn:control}, consider the following finite-horizon centralized optimal control problem:
    \begin{equation}
        \begin{alignedat}{2}
            &\minimize  && \mathbb{E}\left( \textstyle c_T^\lambda(\mathbf{x}_T) + \sum_{t = 0}^{T - 1} c_t^\lambda(\mathbf{x}_t) + c_t^u(\mathbf{u}_t)\right) \\
            &\subjectto && \mathbf{x}_{t + 1} = \tilde{A}_t \mathbf{x}_t + \tilde{B}_t \mathbf{u}_t + \mathbf{w}_{t + 1}, \\
            &           && \mathbf{u}_t \in \mathcal{L}^2(\mathcal{F}_t), \quad t \in \mathbb{T},
        \end{alignedat}
    \end{equation}
    where \(c_t^\lambda(\mathbf{x}_t) \triangleq \mathbf{x}_t^\top \tilde{Q}_t^\lambda \mathbf{x}_t + \mathbf{x}_t^\top {\mathbf{b}_t^\lambda}\), \(\tilde{Q}_t^\lambda \triangleq I_k \otimes (Q_t + Q_t^\lambda) + E_k \otimes P_t\), \(Q_t^\lambda \triangleq 4 \lambda Q_t \Sigma Q_t\), \(\mathbf{b}_t^\lambda \triangleq \mathbf{1}_k \otimes b_t^\lambda\), and \(b_t^\lambda \triangleq 4 \lambda Q_t \gamma_t\).
\end{problem}

\begin{proposition}
    \label{prop:equal_problems}
    If \(\mathbf{u}^*\) is an optimal control for \Cref{prob:mean_field_optimal_control_problem} (resp., \Cref{prob:centralized_optimal_control_problem}), then \(\mathbf{u}^*\) is also an optimal control for \Cref{prob:centralized_optimal_control_problem} (resp., \Cref{prob:mean_field_optimal_control_problem}).
\end{proposition}
\begin{proof}
    Recall from the discussion of \Cref{assn:control} that \(u_t^i \in \mathcal{L}^2(\mathcal{F}_t)\) for each \(i \in \mathbb{I}\) is equivalent to \(\mathbf{u}_t \in \mathcal{L}^2(\mathcal{F}_t)\). By additionally applying \Cref{prop:centralized_dynamics}, the constraints of the two problems are equivalent. Then, we use  \Cref{prop:centralized_costs,prop:centralized_cost_risk_aware} and observe that \(\tilde{Q}_t^\lambda = \tilde{Q}_t + 4 \lambda (I_k \otimes Q_t \Sigma Q_t) = \tilde{Q}_t + I_k \otimes Q_t^\lambda \) to conclude that the objectives of the two problems only differ by a constant.
\end{proof}

We have now reduced the risk-aware finite-horizon optimal control problem (\Cref{prob:mean_field_optimal_control_problem}) to a problem similar to a standard LQR (\Cref{prob:centralized_optimal_control_problem}), except with an extra affine term in the objective. Instead of solving \Cref{prob:centralized_optimal_control_problem} immediately, however, we observe that the matrices in the centralized formulation admit a particularly useful form. We investigate properties of such matrices next, which will facilitate the process of determining the optimal mean-field coupled control.

\section{Pseudo-Block Diagonal Matrices}
\label{sec:pseudo_block_diagonal_matrices}

For any two matrices \(M\) and \(\bar{M}\) of the same dimensions, consider a matrix of the form (recall that \(E_k = \frac{1}{k} \mathbf{1}_k \mathbf{1}_k^\top\))
\begin{equation}
    \label{eq:pseudo_block_diagonal_form}
    \tilde{M} = I_k \otimes M + E_k \otimes (\bar{M} - M).
\end{equation}
If \(M = \bar{M}\), then \(\tilde{M}\) is simply block diagonal. However, when this is not the case, we can view \(\tilde{M}\) as \emph{pseudo-block diagonal}, i.e., a block diagonal matrix perturbed by another matrix with all blocks identical. Such matrices have useful properties that ``decouple'' the structures of \(M\) and \(\bar{M}\) in some sense, induced by the properties of \(E_k\) and the Kronecker product. First, we prove some useful properties of \(E_k\).

\begin{lemma}
    \label{lemma:e_k_properties}
    For any \(k \in \mathbb{N}\), the following four properties hold: (a) \( E_k^2 = E_k \); (b)   \(E_k^\top = E_k\); (c)    \( E_k \mathbf{1}_k = \mathbf{1}_k \) and \(\mathbf{1}_k^\top E_k = \mathbf{1}_k^\top\); and (d)  \( E_k \in \mathcal{S}^k_{++}\).
\end{lemma}
\begin{proof}
    Properties (a)--(c) follow directly from algebraic simplification after substituting the definition \(E_k = \frac{1}{k} \mathbf{1}_k \mathbf{1}_k^\top\). Since \(E_k\) is symmetric and for any nonzero \(x \in \mathbb{R}^k\), we have \(\textstyle x^\top E_k x = \frac{1}{k} x^\top \mathbf{1}_k \mathbf{1}_k^\top x = \frac{1}{k} (x^\top \mathbf{1}_k)^2 > 0 \), we conclude that \(E_k \in \mathcal{S}^k_{++}\).
\end{proof}

The properties in \Cref{lemma:e_k_properties} also hold for \(I_k\), therefore \(E_k\) behaves like a ``pseudo-identity'' matrix. Next, we review properties of the Kronecker product for the reader's convenience.

\begin{lemma}
    \label{lemma:kronecker_properties}
    For any matrices \(A\), \(B\), \(C\), and \(D\) of appropriate dimensions and \(\lambda \in \mathbb{R}\), the following properties hold:
    \begin{subequations}
        \begin{align}
            \label{eq:kronecker_transpose}
            (A \otimes B)^\top &= A^\top \otimes B^\top, \\
            \label{eq:kronecker_multiplication_scalar}
            (\lambda A) \otimes B &= \lambda (A \otimes B) = A \otimes (\lambda B), \\
            \label{eq:kronecker_addition_right}
            A \otimes (B + C) &= A \otimes B + A \otimes C, \\
            \label{eq:kronecker_addition_left}
            (A + B) \otimes C &= A \otimes C + B \otimes C, \\
            \label{eq:kronecker_multiplication}
            (A \otimes B) (C \otimes D) &= A C \otimes B D, \\
            \label{eq:kronecker_inverse}
            (A \otimes B)^{-1} &= A^{-1} \otimes B^{-1},
        \end{align}
    \end{subequations}
    where \cref{eq:kronecker_inverse} holds if and only if \(A\) and \(B\) are invertible.
\end{lemma}
\begin{proof}
    All the properties above follow from straightforward algebraic manipulation (e.g., see \cite[Ch. 2]{hardy2019matrix} for details).
\end{proof}

Now, we define a family of maps from a pair of matrices to a pseudo-block diagonal matrix.
\begin{definition}
    For each \(k\), \(m\), and \(n\), all natural numbers, we define \(\varphi_k^{m , n} : \mathbb{R}^{m \times n} \times \mathbb{R}^{m \times n} \to \mathbb{R}^{m k \times n k}\) as \(\varphi_k^{m , n}(A, \bar{A}) \triangleq I_k \otimes A + E_k \otimes (\bar{A} - A)\).
\end{definition}
We  allow slight abuse of notation by writing \(\varphi_k\) instead of \(\varphi_k^{m , n}\), where \(m\) and \(n\) are understood from the inputs. \(\varphi_k\) presents a convenient representation of pseudo-block diagonal matrices and simplifies many computations because of the properties it inherits from \(E_k\) and the Kronecker product.

\begin{proposition}
    \label{prop:pseudo_diagonal_properties}
    For any matrices \(A\), \(\bar{A}\), \(B\), \(\bar{B}\), \(C\), and \(\bar{C}\) of appropriate dimensions, column vector \(v\) of appropriate dimension, and \(\lambda \in \mathbb{R}\), the following properties hold:
    \begin{subequations}
        \begin{align}
            \label{eq:pseudo_diagonal_transpose}
            \varphi_k(A, \bar{A})^\top &= \varphi_k(A^\top, \bar{A}^\top), \\
            \label{eq:pseudo_diagonal_multiplication_scalar}
            \lambda \varphi_k(A, \bar{A}) &= \varphi_k(\lambda A, \lambda \bar{A}), \\
            \label{eq:pseudo_diagonal_addition}
            \varphi_k(A, \bar{A}) + \varphi_k(B, \bar{B}) &= \varphi_k(A + B, \bar{A} + \bar{B}), \\
            \label{eq:pseudo_diagonal_multiplication_vector}
            \varphi_k(A, \bar{A})(\mathbf{1}_k \otimes v) &= \mathbf{1}_k \otimes \bar{A} v, \\
            \label{eq:pseudo_diagonal_multiplication}
            \varphi_k(A, \bar{A}) \varphi_k(C, \bar{C}) &= \varphi_k(A C, \bar{A} \bar{C}), \\
            \label{eq:pseudo_diagonal_inverse}
            (\varphi_k(A, \bar{A}))^{-1} &= \varphi_k(A^{-1}, \bar{A}^{-1}),
        \end{align}
    \end{subequations}
    where the last line holds if \(A\) and \(\bar{A}\) are invertible.
\end{proposition}
\begin{proof}
    \Cref{eq:pseudo_diagonal_transpose,eq:pseudo_diagonal_multiplication_scalar,eq:pseudo_diagonal_addition} follow from \Cref{lemma:kronecker_properties}. \Cref{eq:pseudo_diagonal_multiplication_vector} can be verified by applying \Cref{lemma:e_k_properties,lemma:kronecker_properties}:
    \begin{align*}
        \varphi_k(A, \bar{A}) (\mathbf{1}_k \otimes v)
        &=  I_k \mathbf{1}_k \otimes A v + E_k \mathbf{1}_k \otimes (\bar{A} - A) v \\
        &=  \mathbf{1}_k \otimes A v + \mathbf{1}_k \otimes (\bar{A} - A) v = \mathbf{1}_k \otimes \bar{A} v.
    \end{align*}
    \Cref{eq:pseudo_diagonal_multiplication} also follows from the previous lemmas:
    \begin{align*}
        & \varphi_k(A, \bar{A}) \varphi_k(C, \bar{C}) \\
        &\hspace{1mm}=  I_k \otimes A C + E_k \otimes ( (\bar{A} - A) C + A (\bar{C} - C) ) \\ &\hspace{1mm}\phantom{{}=} + E_k^2 \otimes (\bar{A} - A) (\bar{C} - C) \\
        &\hspace{1mm}=  I_k \otimes A C \\ &\phantom{{}=} + E_k \otimes ((\bar{A} - A) C + A (\bar{C} - C) + (\bar{A} - A) (\bar{C} - C)) \\
        &\hspace{1mm}=  I_k \otimes A C + E_k \otimes (\bar{A} \bar{C} - A C).
    \end{align*}
    Then, \cref{eq:pseudo_diagonal_inverse} follows from using \cref{eq:pseudo_diagonal_multiplication} to multiply by the proposed inverse and showing that the product is the identity.
\end{proof}

To see the utility of \Cref{prop:pseudo_diagonal_properties}, we observe that the dynamics and cost matrices in the centralized reformulation admit a pseudo-block diagonal form and can be written in terms of \(\varphi_k\).

\begin{lemma}
    \label{lemma:centralized_matrices_phi_k_form}
    We have \(\tilde{A}_t = \varphi_k(A_t, \bar{A}_t)\), \(\tilde{B}_t = \varphi_k(B_t, B_t)\), \(\tilde{Q}_t^\lambda = \varphi_k(Q_t + Q_t^\lambda, \bar{Q}_t + Q_t^\lambda)\), and \(\tilde{R}_t = \varphi_k(R_t, R_t)\) for each \(t\), where \(\bar{A}_t \triangleq A_t + C_t\), \(\bar{Q}_t \triangleq P_t + Q_t\), and the remaining matrices are as in the dynamics and costs defined earlier.
\end{lemma}

The results in this section allow us to view the centralized matrices in a decoupled manner. Given any pseudo-block diagonal matrix \(\tilde{M} = \varphi_k(M, \bar{M})\) in \Cref{lemma:centralized_matrices_phi_k_form} encoding information about the centralized system, \(M\) and \(\bar{M}\) encode the corresponding information about a subsystem with state \(x_t^i - \bar{x}_t\) and the (mean-field) subsystem with state \(\bar{x}_t\), respectively; e.g., \(\tilde{A}_t = \varphi_k(A_t, \bar{A}_t)\), \(A_t\), and \(\bar{A}_t\) are the state update matrices for the centralized system, a subsystem with state \(x_t^i - \bar{x}_t\), and the mean-field subsystem with state \(\bar{x}_t\), respectively.

\section{Optimal Control}
\label{sec:optimal_control}

We present the centralized optimal controller and use the findings from \Cref{sec:pseudo_block_diagonal_matrices} to derive its mean-field coupled form.
\begin{theorem}[Centralized Optimal Control]
    \label{thm:centralized_control}
    The unique centralized optimal control solution to \Cref{prob:centralized_optimal_control_problem} for \(t \in \mathbb{T}\) is
    \begin{equation}\label{eq:centralized_optimal_control}
        \mathbf{u}_t^* = \tilde{K}_t \mathbf{x}_t + \mathbf{f}_t,
    \end{equation}
    where the feedback matrices and affine terms are
    \begin{align}
        \tilde{K}_t &\triangleq - (\tilde{R}_t + \tilde{B}_t^\top \tilde{S}_{t + 1} \tilde{B}_t)^{-1} \tilde{B}_t^\top \tilde{S}_{t + 1} \tilde{A}_t, \\
        \mathbf{f}_t &\triangleq \textstyle - (\tilde{R}_t + \tilde{B}_t^\top \tilde{S}_{t + 1} \tilde{B}_t)^{-1} \tilde{B}_t^\top (\tilde{S}_{t + 1} \bm{\mu} + \frac{1}{2} \mathbf{g}_{t + 1}), \\
        \intertext{and where \(\tilde{S}_t\) and \(\mathbf{g}_t\) are}
        \begin{split}
                \tilde{S}_t
            &\triangleq  - \tilde{A}_t^\top \tilde{S}_{t + 1} \tilde{B}_t (\tilde{R}_t + \tilde{B}_t^\top \tilde{S}_{t + 1} \tilde{B}_t)^{-1} \tilde{B}_t^\top \tilde{S}_{t + 1} \tilde{A}_t \\
            &\phantom{{}=} + \tilde{A}_t^\top \tilde{S}_{t + 1} \tilde{A}_t + \tilde{Q}_t^\lambda,
        \end{split} \\
        \mathbf{g}_t &\triangleq (\tilde{A}_t + \tilde{B}_t \tilde{K}_t)^\top (2 \tilde{S}_{t + 1} \bm{\mu} + \mathbf{g}_{t + 1}) + \mathbf{b}_t^\lambda,
    \end{align}
    with \(\mathbf{g}_T \triangleq \mathbf{b}_T^\lambda\), \(\tilde{S}_T \triangleq \tilde{Q}_T^\lambda\), and \(\bm{\mu} \triangleq \mathbf{1}_k \otimes \mu\).
\end{theorem}
\begin{proof}
    This result can be shown by applying dynamic programming (e.g., see \cite[Ch. 6]{kumar2015stochastic}). The optimal value function takes the form \(V_t(\mathbf{x}_t) = \mathbf{x}_t^\top \tilde{S}_t \mathbf{x}_t + \mathbf{x}_t^\top \mathbf{g}_t + r_t\), where \(\tilde{S}_t \in \mathcal{S}_+^{nk}\) and \(r_t \in \mathbb{R}\). This is standard, so we omit the derivation.
\end{proof}

We now apply the properties in \Cref{prop:pseudo_diagonal_properties} to decompose the previous solution into a mean-field coupled form.

\begin{theorem}[Mean-Field Coupled Optimal Control]
    \label{thm:mean_field_coupled_control}
    \newcounter{customsubequation}
    \renewcommand{\thecustomsubequation}{\theequation\alph{customsubequation}}
    The unique mean-field coupled optimal control solution to \Cref{prob:mean_field_optimal_control_problem} for \(t \in \mathbb{T}\) is
    \begin{align}\label{eq:mean_field_coupled_optimal_control}
        {u_t^i}^* = K_t x_t^i + (\bar{K}_t - K_t) \bar{x}_t + f_t, \quad i \in \mathbb{I},
    \end{align}
    where the feedback matrices and affine terms are
    \begin{align}
        \setcounter{customsubequation}{0}\refstepcounter{equation}
        \refstepcounter{customsubequation}\tag{\thecustomsubequation}
        K_t &\triangleq - (R_t + B_t^\top S_{t + 1} B_t)^{-1} B_t^\top S_{t + 1} A_t, \\
        \refstepcounter{customsubequation}\tag{\thecustomsubequation}
        \bar{K}_t &\triangleq - (R_t + B_t^\top \bar{S}_{t + 1} B_t)^{-1} B_t^\top \bar{S}_{t + 1} \bar{A}_t, \\
        f_t &\triangleq \textstyle - (R_t + B_t^\top \bar{S}_{t + 1} B_t)^{-1} B_t^\top (\bar{S}_{t + 1} \mu + \frac{1}{2} g_{t + 1}), \\
        \intertext{and where \(S_t\), \(\bar{S}_t\), and \(g_t\) are}
        \setcounter{customsubequation}{0}\refstepcounter{equation}
        \refstepcounter{customsubequation}\tag{\thecustomsubequation}
        \begin{split}
            S_t &\triangleq - A_t^\top S_{t + 1} B_t (R_t + B_t^\top S_{t + 1} B_t)^{-1} B_t^\top S_{t + 1} A_t \\ &\phantom{{}=} + A_t^\top S_{t + 1} A_t + Q_t + Q_t^\lambda,
        \end{split} \\
        \refstepcounter{customsubequation}\tag{\thecustomsubequation}
        \begin{split}
            \bar{S}_t &\triangleq - \bar{A}_t^\top \bar{S}_{t + 1} B_t (R_t + B_t^\top \bar{S}_{t + 1} B_t)^{-1} B_t^\top \bar{S}_{t + 1} \bar{A}_t \\ &\phantom{{}=} + \bar{A}_t^\top \bar{S}_{t + 1} \bar{A}_t + \bar{Q}_t + Q_t^\lambda,
        \end{split} \\
        g_t &\triangleq (\bar{A}_t + B_t \bar{K}_t)^\top (2 \bar{S}_{t + 1} \mu + g_{t + 1}) + b_t^\lambda,
    \end{align}
    with \(g_T \triangleq b_T^\lambda\), \(S_T \triangleq Q_T + Q_T^\lambda\), \(\bar{S}_T \triangleq \bar{Q}_T + Q_T^\lambda\), \(\bar{A}_t\) and \(\bar{Q}_t\) as in \Cref{lemma:centralized_matrices_phi_k_form}, and \(Q_t^\lambda\) and \(b_t^\lambda\) as in \Cref{prob:centralized_optimal_control_problem}.
\end{theorem}
\begin{proof}
    Since \(\mathbf{u}^* = (\mathbf{u}_0^*, \mathbf{u}_1^*, \ldots, \mathbf{u}_{T-1}^*)\) \cref{eq:centralized_optimal_control}  is optimal for \Cref{prob:centralized_optimal_control_problem} by \Cref{thm:centralized_control}, \(\mathbf{u}^*\) is also optimal for \Cref{prob:mean_field_optimal_control_problem} by \Cref{prop:equal_problems}. Recalling that \(\mathbf{u}_t^* = ({u_t^1}^*, \ldots, {u_t^k}^*)\), our task is to extract \({u_t^i}^*\) from \(\mathbf{u}_t^*\). We can prove that \(\tilde{S}_t = \varphi_k(S_t, \bar{S}_t)\), where \(S_t \in \mathcal{S}_+^n\) and \(\bar{S}_t \in \mathcal{S}_+^n\), for each \(t \in \tilde{\mathbb{T}}\) via backwards induction. For the base case, observe that \(\tilde{S}_T = \tilde{Q}_T^\lambda = \varphi_k(Q_T + Q_T^\lambda, \bar{Q}_T + Q_T^\lambda) = \varphi_k(S_T, \bar{S}_T)\) using \Cref{thm:centralized_control} and \Cref{lemma:centralized_matrices_phi_k_form}, where \(S_T \in \mathcal{S}_+^n\) and \(\bar{S}_T \in \mathcal{S}_+^n\). For the inductive step, first assume that \(\tilde{S}_{t + 1} = \varphi_k(S_{t + 1}, \bar{S}_{t + 1})\), where \(S_{t+1} \in \mathcal{S}_+^n\) and \(\bar{S}_{t+1} \in \mathcal{S}_+^n\), for some \(t \in \mathbb{T}\).  Then, \Cref{lemma:centralized_matrices_phi_k_form} and \Cref{prop:pseudo_diagonal_properties} (with invertibility of certain matrices) lead to
    \begin{align*}
        & (\tilde{R}_t + \tilde{B}_t^\top \tilde{S}_{t + 1} \tilde{B}_t)^{-1} \\
        &\hspace{1mm}=  (\tilde{R}_t + \varphi_k(B_t, B_t)^\top \varphi_k(S_{t + 1}, \bar{S}_{t + 1}) \varphi_k(B_t, B_t))^{-1} \\
        &\hspace{1mm}=  (\tilde{R}_t + \varphi_k(B_t^\top, B_t^\top) \varphi_k(S_{t + 1} B_t, \bar{S}_{t + 1} B_t))^{-1} \\
        &\hspace{1mm}=  (\varphi_k(R_t, R_t) + \varphi_k(B_t^\top S_{t + 1} B_t, B_t^\top \bar{S}_{t + 1} B_t))^{-1} \\
        &\hspace{1mm}=  (\varphi_k(R_t + B_t^\top S_{t + 1} B_t, R_t + B_t^\top \bar{S}_{t + 1} B_t))^{-1} \\
        &\hspace{1mm}=  \varphi_k((R_t + B_t^\top S_{t + 1} B_t)^{-1}, (R_t + B_t^\top \bar{S}_{t + 1} B_t)^{-1}).
    \end{align*}
    For space considerations, we omit the rest of the inductive step, in which we can similarly apply the properties in \Cref{prop:pseudo_diagonal_properties} to conclude that \(\tilde{S}_t = \varphi_k(S_t, \bar{S}_t)\) and also show that \(S_t\) and \(\bar{S}_t\) are symmetric positive semidefinite.

    Next, we use \(\tilde{S}_t = \varphi_k(S_t, \bar{S}_t)\) to verify that \(\tilde{K}_t = \varphi_k(K_t, \bar{K}_t)\), again by applying the same properties. We can also show that \(\mathbf{g}_t = \mathbf{1}_k \otimes g_t\) for each \(t \in \tilde{\mathbb{T}}\) by induction, where we use \cref{eq:pseudo_diagonal_multiplication_vector} in particular. For each \(t \in \mathbb{T}\), the equality \(\mathbf{f}_t = \mathbf{1}_k \otimes f_t\) follows from \(\tilde{S}_{t+1} = \varphi_k(S_{t+1}, \bar{S}_{t+1})\), \(\mathbf{g}_{t+1} = \mathbf{1}_k \otimes g_{t+1}\), \Cref{lemma:centralized_matrices_phi_k_form}, and \Cref{prop:pseudo_diagonal_properties}. Finally, we observe the block row of \(\mathbf{u}_t^* = \tilde{K}_t \mathbf{x}_t + \mathbf{f}_t = \varphi_k(K_t, \bar{K}_t) \mathbf{x}_t + \mathbf{1}_k \otimes f_t\) corresponding to the \(i\)th subsystem to find that \({u_t^i}^* = \textstyle K_t x_t^i + \frac{1}{k} \sum_{j = 1}^k (\bar{K}_t - K_t) x_t^j + f_t\), which is equivalent to \cref{eq:mean_field_coupled_optimal_control}.
\end{proof}

The solution pathway in this section is not specific to predictive variance. In particular, a similar pathway may be used to solve an optimal control problem for mean-field coupled subsystems with an alternative risk formulation, such as the exponential utility, by first reformulating the problem in terms of pseudo-block diagonal matrices and then following an approach similar to the proof of \Cref{thm:mean_field_coupled_control}.

Note that when \(\lambda = 0\) and the disturbance is zero-mean, the optimal control in \Cref{thm:mean_field_coupled_control} is identical to the (risk-neutral) state-feedback optimal control in \cite[Th. 1]{arabneydi2015teamoptimal}. Also, the optimal control in \Cref{thm:mean_field_coupled_control} is analogous to the one in \cite[Th. 1]{roudneshin2023risk}, but without the redundant affine term evaluating to zero.

\section{Numerical Example}
\label{sec:numerical_example}

\begin{figure*}[htbp]
    \centering

    \subfloat[Average Subsystem State Energy vs. \(t\)]{%
        \includegraphics[width=0.44\linewidth]{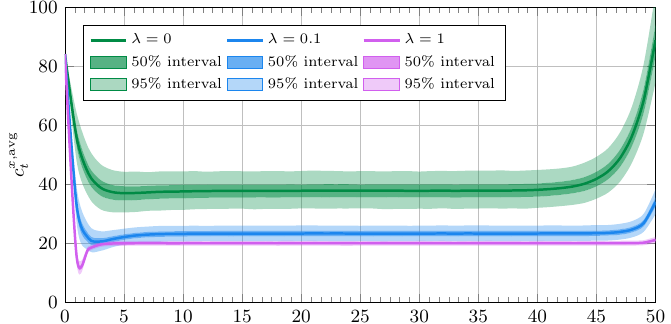}%
        \label{fig:avg_state_energy_vs_t}%
    }%
    \hspace*{0.06\linewidth}%
    \subfloat[Maximum Subsystem State Energy vs. \(t\)]{%
        \includegraphics[width=0.44\linewidth]{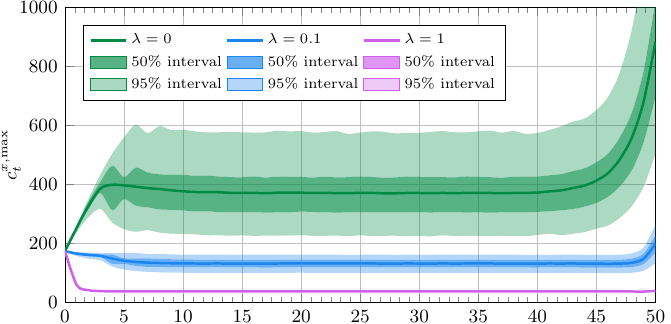}%
        \label{fig:max_state_energy_vs_t}%
    } \\

    \subfloat[Time Average of State Energy Statistics vs. \(\lambda\)]{%
        \includegraphics[width=0.44\linewidth]{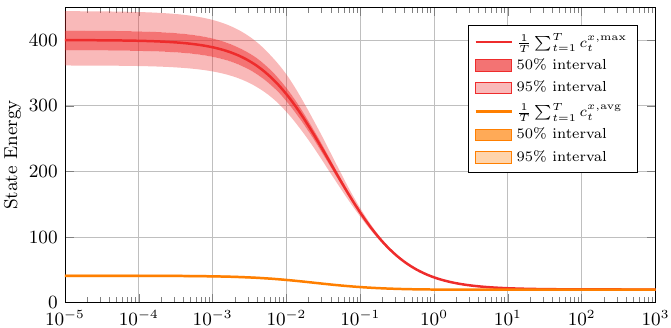}%
        \label{fig:state_energy_vs_lambda}%
    }%
    \hspace*{0.06\linewidth}%
    \subfloat[Time Average of Control Effort Statistics vs. \(\lambda\)]{%
        \includegraphics[width=0.44\linewidth]{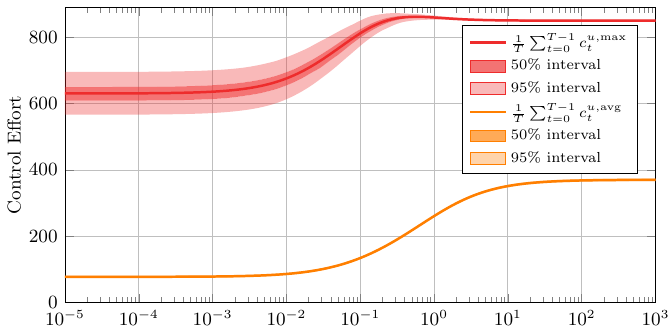}%
        \label{fig:control_effort_vs_lambda}%
    }

    \caption{Plots of \protect\subref{fig:avg_state_energy_vs_t} average state energy \(c_t^{x, \mathrm{avg}} \triangleq \frac{1}{k} \sum_{i \in \mathbb{I}} (x_t^i)^\top Q_t x_t^i\) vs. \(t\), \protect\subref{fig:max_state_energy_vs_t} maximum state energy \(c_t^{x, \mathrm{max}} \triangleq \max_{i \in \mathbb{I}} (x_t^i)^\top Q_t x_t^i\) vs. \(t\), \protect\subref{fig:state_energy_vs_lambda} time average of state energy statistics, \(c_t^{x, \mathrm{avg}}\) and \(c_t^{x, \mathrm{max}}\), vs. \(\lambda\), and \protect\subref{fig:control_effort_vs_lambda} time average of control effort statistics, \(c_t^{u, \mathrm{avg}} \triangleq \frac{1}{k} \sum_{i \in \mathbb{I}} (u_t^i)^\top R_t u_t^i\) and \(c_t^{u, \mathrm{max}} \triangleq \max_{i \in \mathbb{I}} (u_t^i)^\top R_t u_t^i\), vs. \(\lambda\). The horizontal axis labels are \protect\subref{fig:avg_state_energy_vs_t} time \(t\), \protect\subref{fig:max_state_energy_vs_t} time \(t\), \protect\subref{fig:state_energy_vs_lambda} \(\lambda\), and \protect\subref{fig:control_effort_vs_lambda} \(\lambda\). The bold lines and shaded regions indicate the empirical mean values and observed credible intervals, respectively, of the plotted quantities among the \(10^4\) simulations.}
    \label{fig:numerical_sim_plots}
\end{figure*}

Consider \(k = 250\) mean-field coupled LQ subsystems over a time horizon of length \(T = 50\) with parameters \(A_t = 1.1\), \(B_t = 0.3\), \(C_t = 0.2\), \(P_t = 0.4\), \(Q_t = 0.8\), and \(R_t = 1.2\), with the deterministic initial states \(x_0^1, \ldots, x_0^k\) being fixed realizations from \(\mathcal{N}(10, 2)\), and with \(w_t^i \sim 10 (\operatorname{Bernoulli}(0.25) - 0.25)\), a simple example of a random variable with nonzero skew. We run \(10^4\) simulations, each with an independently generated disturbance process, and then plot statistics of average and maximum subsystem state energy and control effort in \Cref{fig:numerical_sim_plots} (see caption for definitions).

In \Cref{fig:avg_state_energy_vs_t,fig:max_state_energy_vs_t}, the empirical means of \(c_t^{x, \mathrm{avg}}\) and \(c_t^{x, \mathrm{max}}\) decrease as \(\lambda\) increases, indicating improved regulation of \emph{all} subsystems in general, rather than just the mean-field (recall \cref{eq:mean_field_costs}). The substantial decrease in \(c_t^{x, \mathrm{max}}\) in particular shows that increasing \(\lambda\) can provide regulation for even the worst-case subsystem. In contrast, the worst-case subsystem in the risk-neutral case may have no regulation, with its state energy \emph{increasing} significantly over time.

In \Cref{fig:state_energy_vs_lambda,fig:control_effort_vs_lambda}, the empirical variability in the time averages of \(c_t^{x, \mathrm{max}}\) and \(c_t^{u, \mathrm{max}}\) decrease as \(\lambda\) increases, demonstrating the controller's robustness by making the performance of the risk-aware controller highly predictable regardless of the specific realization of the disturbance process. Furthermore, with increasing \(\lambda\), \(c_t^{x, \mathrm{max}}\) decreases dramatically, while \(c_t^{u, \mathrm{max}}\) increases less significantly. However, \(c_t^{x, \mathrm{avg}}\) decreases marginally compared to the increase in \(c_t^{u, \mathrm{avg}}\), indicating an increase in average total cost compared to the risk-neutral case. So, \(\lambda\) presents a trade-off to sacrifice average performance for less volatility and improved worst-case performance.

\section{Conclusion}
\label{sec:conclusion}

We present and solve a risk-aware social optimal control problem for mean-field coupled subsystems using predictive variance. Our solution pathway involves reformulating the problem in terms of pseudo-block diagonal matrices, which enjoy many notable properties not specific to predictive variance. This pathway is useful for extending existing solutions for uncoupled systems to mean-field coupled settings, and in particular, may help extend the results of this work to a partially observable setting, e.g., see \cite{koumpis2022stateoutput}. Another direction of interest is to explore alternative risk or cost formulations which may admit a centralized formulation with pseudo-block diagonal matrices, such as distributionally robust constraints \cite{vanParys2016distributionally} or variance suppression \cite{fujimoto2011optimal}.

\section*{Acknowledgment}

The authors thank Dr. Shuang Gao for fruitful discussions.


\end{document}